\documentclass[onecolumn,english,conference]{IEEEtran}
\pdfoutput=1
\usepackage[T1]{fontenc}
\usepackage[active]{srcltx}
\usepackage{color}
\usepackage{babel}
\usepackage{array}
\usepackage{float}
\usepackage{textcomp}
\usepackage{mathrsfs}
\usepackage{multirow}
\usepackage{amsmath}
\usepackage{amssymb}
\usepackage{stackrel}
\usepackage{graphicx}
\usepackage[unicode=true,
 bookmarks=true,bookmarksnumbered=true,bookmarksopen=true,bookmarksopenlevel=1,
 breaklinks=false,pdfborder={0 0 0},pdfborderstyle={},backref=false,colorlinks=false]
 {hyperref}
\hypersetup{pdftitle={Your Title},
 pdfauthor={Your Name},
 pdfpagelayout=OneColumn, pdfnewwindow=true, pdfstartview=XYZ, plainpages=false}

\makeatletter

\providecommand{\tabularnewline}{\\}
\floatstyle{ruled}
\newfloat{algorithm}{tbp}{loa}
\providecommand{\algorithmname}{Algorithm}
\floatname{algorithm}{\protect\algorithmname}

\IEEEoverridecommandlockouts
\ifCLASSOPTIONcompsoc
\usepackage[caption=false,font=normalsize,labelfont=sf,textfont=sf]{subfig}
\else
\usepackage[caption=false,font=footnotesize]{subfig}
\fi
\usepackage{cite}
\usepackage{amsthm}

\newtheorem{proposition}{\textbf{Proposition}}
\usepackage{algorithm}
\usepackage{amsmath}

\usepackage{cite}
\usepackage{bm}
\usepackage{algorithmic}
\usepackage{algorithm}
\usepackage{graphicx}
\renewcommand{\fnum@figure}{Fig.~\thefigure}
\IEEEoverridecommandlockouts

\makeatother

\begin{document}
\title{Cooperative Tri-Point Model-Based Ground-to-Air Coverage Extension
	in Beyond 5G Networks}
\author{\IEEEauthorblockN{Ziwei Cai, Min Sheng, Junyu Liu, Chenxi Zhao, and Jiandong Li}\\
	\IEEEauthorblockA{State Key Laboratory of Integrated Service Networks,
		Xidian University, Xi'an, Shaanxi, 710071, China\\
}}

\maketitle
\begin{abstract}
The utilization of existing terrestrial infrastructures to provide
coverage for aerial users is a potentially low-cost solution. However,
the already deployed terrestrial base stations (TBSs) result in weak
ground-to-air (G2A) coverage due to the down-tilted antennas. Furthermore,
	achieving optimal coverage across the entire airspace through antenna
	adjustment is challenging due to the complex signal coverage requirements in three-dimensional space, especially in the vertical direction. In this paper, we propose a cooperative
tri-point (CoTP) model-based method that utilizes cooperative beams
to enhance the G2A coverage extension. To utilize existing TBSs for establishing
effective cooperation, we prove that the cooperation among three TBSs
can ensure G2A coverage with a minimum coverage overlap, and design
the CoTP model  to analyze the G2A coverage
extension. Using the model, a cooperative coverage structure
based on Delaunay triangulation is designed to divide triangular prism-shaped
subspaces and corresponding TBS cooperation sets. To enable TBSs in the cooperation
set to cover different height subspaces while maintaining ground coverage,
we design a cooperative beam generation algorithm to maximize the
coverage in the triangular prism-shaped airspace. The simulation results
and field trials demonstrate that the proposed method can efficiently
enhance the G2A coverage extension while guaranteeing
ground coverage.

$\vphantom{}$

\end{abstract}

\section{introduction}
\label{s1}

With the popularity of aerial applications, including public safety
	patrol, cargo transport, and traffic monitoring, aerial user equipments
(AUEs) have been soaring globally over recent years \cite{ZengYong-uav-potential-challenge, Admission-Control}.
	Enabling AUEs-centric applications requires ubiquitous wireless connectivity
	in beyond-5G (fifth generation) networks. The terrestrial cellular network is a natural
	candidate for serving AUEs \cite{Amer}, since it can take advantage
	of existing ground infrastructures and coverage technologies to provide
	radio access services at low cost. However, it is challenging to efficiently extend coverage from the
ground to three-dimensional (3D) airspace by leveraging the current
terrestrial infrastructure due to the down-tilted antennas of
terrestrial base stations (TBSs) fail to meet the coverage demand in the vertical direction \cite{Cherif-ground-user-only}. Although antennas can be tilted upward,
it is still difficult to design the optimal beam to efficiently cover
3D airspace since it is a multi-objective optimization problem that
involves maximizing coverage while minimizing overlap \cite{Maeng}. Therefore,
	we intend to provide answers to the following questions for ground-to-air (G2A) coverage
	extension: 1) how to design the cooperative coverage structure based
	on the deployed TBSs, and 2) how to optimize antenna beam parameters
	of cooperative TBSs to implement cooperative airspace coverage.

\subsection{Related Work and Motivation}

The coverage enhancement methods of terrestrial cellular networks
have been widely studied over the past decades \cite{Liuyaxi,Balevi-beamforming-coverage-maxi,Traditional-Cell}. These studies have focused on investigating methods to optimize ground coverage by adjusting antennas. However, they ignore the coverage demand in the vertical direction, which is an important feature of the G2A
coverage extension \cite{Azari-G2A-different}. Regarding the 
G2A coverage, the use of a down-tilted antenna may lead to coverage holes. This is attributed to the fact that the down-tilted antenna primarily serves the airspace above it through its side lobes \cite{Cherif-beam-sidelobe, Capacity-evaluation}.

Considering the importance and distinctiveness of the G2A
	coverage, recent studies are exploring ways to enhance coverage for
	airspace. The first study attempting to provide
a comprehensive analysis of the G2A coverage extension can be traced
back to \cite{Azari-first-paper}, which outlines a potential pathway for achieving airspace coverage. 
	The authors of \cite{Mozaffari-6G-technogy} have investigated the
	potential and challenges of G2A coverage extension, highlighting that
	advanced antenna technologies have the ability to enhance the coverage of airspace. In \cite{Lyu}, the authors have offered a
link-level coverage analysis of the cellular network that serves both
AUEs and ground users. Meanwhile, the hand-off of AUEs and the coverage probability
of AUEs have been analyzed in \cite{LiYan-G2A-ComP}.
Although these studies provide the theoretical analysis of the link-level coverage ability of airspace based on stochastic
	geometry, they fail to provide a practical solution to enhance the
	G2A coverage extension.


With the advancement of 3D directional antenna technology, adjustable antenna parameters of TBS can be used to enhance the G2A coverage extension.
The ability to improve the coverage performance by adjusting the tilt
angle and beam-width of the antenna has been demonstrated in \cite{Maeng}. However, it fails to consider
	how to fulfill the coverage requirements across the entire airspace with different heights.
	For the deployed TBSs, the challenge in enhancing the G2A coverage stems from extending coverage into the airspace while ensuring ground coverage.

\subsection{Outcomes and Contributions }

Motivated by the above discussions, in this work, we propose a G2A coverage extension method based on
the cooperative tri-point (CoTP) model. The outcomes and contributions
of this paper can be summarized as follows:
\begin{itemize}
	\item  We prove the property of achieving minimum overlap
		through the cooperation coverage among three TBSs. Subsequently, a CoTP
		model is proposed to analyze the G2A coverage extension in existing
		cellular networks. Based on this, a Delaunay triangulation (DT)-based
		cooperative coverage structure is designed to seamlessly divide a given airspace
		into multiple triangular prism (TP)-shaped airspaces and formulate
		corresponding TBS cooperation sets. 
	\item  We design two cooperative beam generation algorithms
		to maximize the coverage ratio of each TP-shaped airspace while controlling
		overlap ratio. Especially, the space-layered beam cooperative (SLBC) coverage
		algorithm is designed to efficiently cover the low, medium, and high-layered
		subspaces by tuning the discrete antenna pattern. Moreover, the adaptive beam cooperative (ABC) coverage algorithm is designed
		to further enhance 3D coverage and prevent power leakage by generating
		beams with continuously adjustable bandwidth. 
	\item We conduct extensive simulations as well as field trials, which demonstrate that the proposed
	method can improve the coverage ratio of 3D airspace by 83\% compared
	to a cellular network that uses down-tilted antenna parameters, while also guaranteeing ground coverage.
\end{itemize}
The following sections of this paper are organized as follows. In
Section \ref{sec:System-Model}, we present the system model and introduce
performance metrics. In Section \ref{sec:Coverage-Maximization-Deployment},
we propose the CoTP-based G2A coverage extension method. The simulations
and field trials are presented in Section \ref{sec:NUMERICAL-RESULT}.
Finally, conclusions are drawn in Section \ref{sec:CONCLUSION}.

\section{System Model\label{sec:System-Model}}


\subsection{Network Model}

We consider a cellular network for G2A coverage extension (see Fig.
\ref{fig:Cooperative-Tri-Point-Network}), which consists of TBSs
(with the same transmit power) denoted by $\Pi_{\textrm{BS}}=\{\textrm{BS}_{i}\mid i=1,2,\cdots I\}$
and sectors denoted by $\Pi_{\textrm{S}}=\{\textrm{S}{}_{ij}\mid i=1,2,\cdots I,j=1,2,\cdots,J\}$,
where $\textrm{S}{}_{ij}$ denotes the $j-$th sector of $i-$th
TBS. The sets of antennas of TBSs are denoted by $\Pi_{\textrm{A}}=\{\textrm{A}{}_{ij}\mid i=1,2,\cdots I,j=1,2,\cdots,J\}$.
Each sector $\textrm{S}{}_{ij}$ is equipped with one 3D directional
antenna $\textrm{A}{}_{ij}$. The horizontal half-power bandwidth
(H-HPBW), vertical half-power bandwidth (V-HPBW), and tilt angle of each antenna are denoted by $\epsilon=\{\Psi,\Phi,\Theta\}$. The 3D Cartesian coordinate of $\mathrm{BS_{\mathit{i}}}$
is denoted as $(x_{\textrm{BS}_{i}},y_{\textrm{BS}_{i}},z_{\textrm{BS}_{i}})$.
The directional antenna beams with different patterns can provide wireless coverage
for the low-altitude 3D airspace $\varOmega$. The maximal height
of $\varOmega$ is denoted by $h_{max}$. 

The 3D continuous airspace $\varOmega$ can be discretized by a regular
	cubic lattice \cite{grid-division}, which is denoted by $\mathbf{v}_{k}=(x_{k},y_{k},z_{k})$.  Specifically, the coverage ratio of the
	center point of lattice represents that of the entire
	lattice. The number of lattices in the airspace
$\varOmega$ is denoted by $N_{t}$. The tilt angle $\Theta$ determines the direction of the directional beam, thereby
affecting the received power of each lattice. 

\begin{figure}[t]
	\centering
	\includegraphics[width=0.42\textwidth]{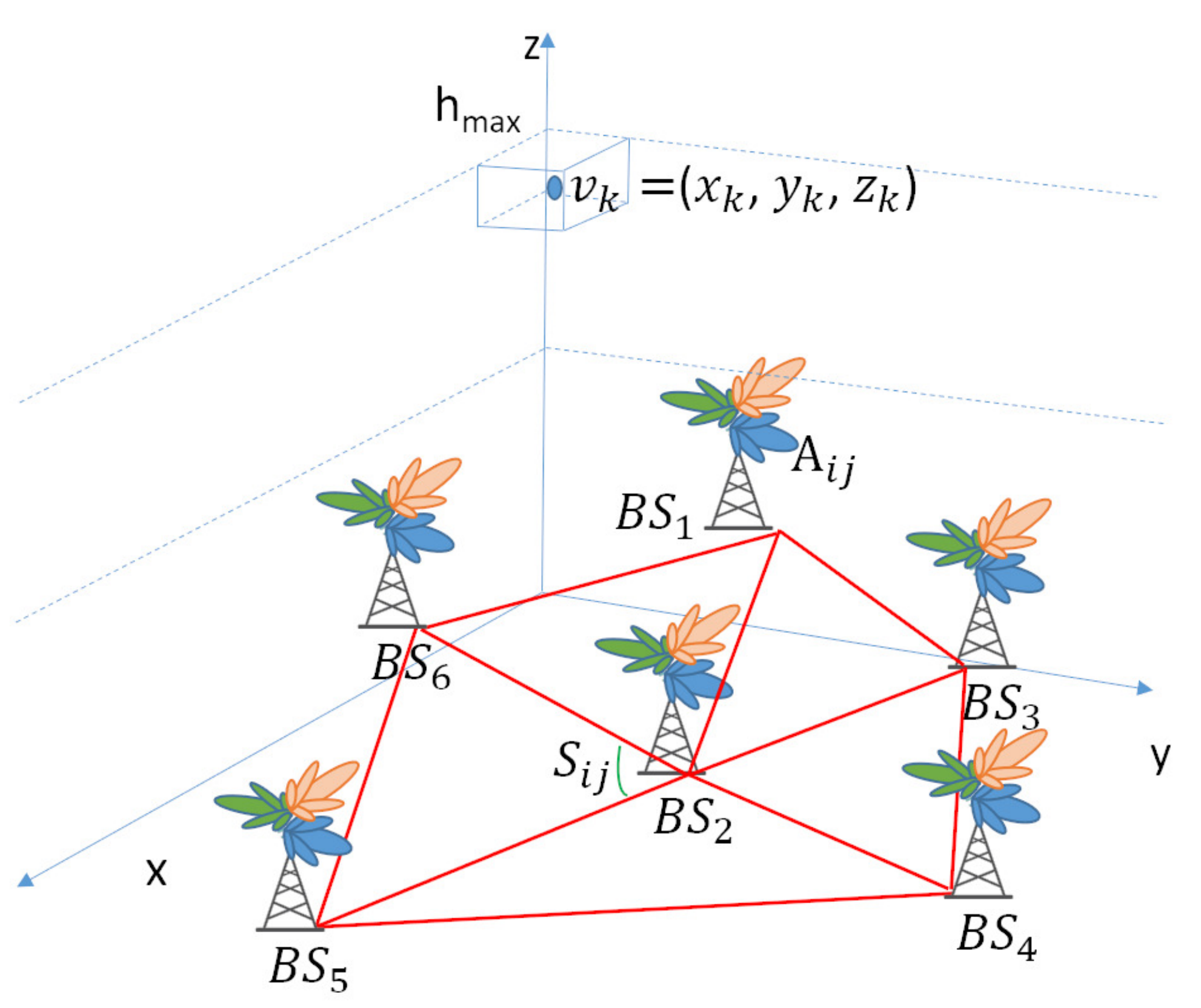}
	\caption{Illustration of the cellular network for G2A coverage extension.}
	\label{fig:Cooperative-Tri-Point-Network}
\end{figure}

\subsection{Channel Model and Antenna Model}
The G2A channel model of lattice $\mathbf{v}_{k}=(x_{k},y_{k},z_{k})$
is given by
\begin{equation}
	PL(d_{2D},h_{t})=p_{\textrm{L}}PL_{\textrm{L}}+p_{\textrm{N}}PL_{\textrm{N}},
\end{equation}
where the two-dimensional distance between the TBS and lattice $\mathbf{v}_{k}$
is given by $d_{2D}=\sqrt{(x_{v_{k}}-x_{\textrm{BS}_{i}})^{2}+(y_{v_{k}}-y_{\textrm{BS}_{i}})^{2}}$,
while the height difference between them is expressed as $h_{t}=\sqrt{(z_{v_{k}}-z_{\textrm{BS}_{i}})^{2}}$.
Moreover, $PL_{\textrm{L}}$ and $PL_{\textrm{N}}$ denote the path loss
in the line-of-sight (LOS) condition and the non-line-of-sight (NLOS)
condition, respectively. $p_{\textrm{L}}$ and $p_{\textrm{N}}$ denote the LOS
probability and the NLOS probability, respectively. The coefficients of the
channel model are given by the 3GPP technical report TR36.777 \cite{TR36.777}.

	The 3D directional antenna gain \cite{Junyu-antenna} is given by
\begin{equation}
		G(\varphi,\phi)\!=\!\begin{cases}
			\mathrm{G}_{0}/(\Psi\!\times\!\Phi),\!\!\!\!\!\! & -\Psi\!\leq\varphi\leq\!\Psi\!,\!-\Phi\!\leq\phi\leq\!\Phi,\\
			\mathrm{S}_{0}, & \mathrm{otherwise},
		\end{cases}\label{eq:antenna model}
\end{equation}
where $\Psi\in(0,\pi)$, $\Phi\in(0,\pi)$, $\mathrm{G}_{0}=2.2864$ and $\mathrm{S}_{0}=0.03$
	\cite{Anten-Side-lobe}.

\subsection{Performance Matrices and Problem Formulation}

The received power of lattice $\mathbf{v}_{k}=(x_{k},y_{k},z_{k})$ is given
by 

\begin{equation}
	P_{\textrm{BS}_{i}}(\mathbf{v}_{k})=P_{\textrm{T}}+G_{\textrm{BS}_{i}}(\varphi_{\mathbf{v}_{k}},\phi_{\mathbf{v}_{k}})+PL_{\textrm{BS}_{i}}(\mathbf{v}_{k}),
\end{equation}
where $P_{\textrm{T}}$ and $G_{\textrm{BS}_{i}}(\varphi_{\mathbf{v}_{k}},\phi_{\mathbf{v}_{k}})$
are the transmit power and the antenna gain of $\textrm{BS}_{i}$,
respectively, and $PL_{\textrm{BS}_{i}}(\mathbf{v}_{k})$ is the path loss of the
G2A channel.

In this paper, we employ the good coverage ratio (GCR) and the coverage overlap
ratio (COR) for evaluating the G2A coverage extension
\cite{Fagen-overlap}. Let $\xi$ denotes the GCR, we have 
\begin{equation}
	\xi=\frac{N_{\textrm{cov}}}{N_{t}},
\end{equation}
where $N_{\textrm{cov}}$ denotes the number of lattices with received
power that exceeds the received power threshold. $N_{\textrm{cov}}$
is given by 
\begin{align}
	N_{\textrm{cov}}=\sum_{k=1}^{N_{t}}C\left(\mathbf{v}_{k}\right),\label{eq:covera area}
\end{align}
where $C\left(\mathbf{v}_{k}\right)$ denotes whether the lattice $\mathbf{v}_{k}$ is covered
by at least one BS. $C\left(\mathbf{v}_{k}\right)$ is given by
\begin{equation}
	C(\mathbf{v}_{k})=\begin{cases}
		1, & P_{\textrm{BS}_{i}}(\mathbf{v}_{k})\geq\tau,\\
		0, & \mathrm{otherwise},
	\end{cases}
\end{equation}
where $\tau$ is the received power threshold. 

Let $\kappa$ denotes the COR, we have
\begin{equation}
	\kappa=\frac{N_{\textrm{over}}}{N_{t}},
\end{equation}
where $N_{\textrm{over}}$ denotes the number of lattices belonging
to the coverage overlap. $N_{\textrm{over}}$ is given by 

\begin{align}
	N_{\textrm{over}}=\sum_{k=1}^{N_{t}}O\left(\mathbf{v}_{k}\right),\label{eq:overlap area}
\end{align}
where $O\left(\mathbf{v}_{k}\right)$ denotes whether the lattice $\mathbf{v}_{k}$ is covered
by more than one BS. $O\left(\mathbf{v}_{k}\right)$ is given by
\begin{equation}
	\!O(v_{k})\!=\!\!\begin{cases}
		\!1,\!\!\!\!\!\! & \exists(\!P_{\textrm{BS}_{i}}(\mathbf{v}_{k})\!\geq\!\tau)\&(P_{\textrm{BS}_{j\neq i}}(\mathbf{v}_{k})\!\geq\!\textrm{\ensuremath{\tau}}\!),\\
		\!0,\!\!\!\!\!\! & \mathrm{otherwise}.
	\end{cases}
\end{equation}

In this paper, we aim at maximizing the coverage for 3D airspace while
controlling overlap. This objective can be transformed into a problem of maximizing coverage subject to overlap constraints. In particular, the optimization
problem can be formulated as 
\begin{align}
	\mathbf{P1}:\underset{\left\{ \Psi,\Phi,\Theta\right\} }{\textrm{max}} & \xi\label{eq:objetive-function}\\
	\textrm{C1}: & 0<\Psi<180^{o},\tag{10a}\\
	\textrm{C2}: & 0<\Phi<180^{o},\tag{10b}\\
	\mathit{\textrm{C3}:} & -90\leq\Theta\leq90^{o},\tag{10c}\\
	\mathit{\textrm{C4}:} & 0\leq\kappa\leq\mathsf{T},\tag{10d}\label{eq:overlap-constraints}
\end{align}
where $\mathsf{T}$ denotes the level of controlled overlap. In $\mathbf{\mathbf{P}1}$, C1-C3 limit the tuning range of the beam parameters, and
C4 limits the coverage overlap ratio of the given 3D airspace. It
is shown that problem $\mathbf{\mathbf{P}1}$ is a non-convex mixed integer
optimization problem.

\section{Coverage Maximization Scheme Based On CoTP Model\label{sec:Coverage-Maximization-Deployment}}

In this section, we prove the minimum coverage overlap property of the cooperation among three TBSs and design the CoTP coverage model. Based on this model,
	a coverage structure based on the Delaunay triangulation (DT) and
	the cooperative beam generation algorithm are proposed for G2A
	coverage extension.

\subsection{CoTP Coverage Model}
It is well-known that the square, triangle, and hexagon are regarded as potential coverage structures in the context of cellular network planning. To extend G2A coverage, three prisms
can be constructed by extending these three planar structures vertically,
including TP, square-prism (SP), and hexagonal-prism (HP). Then, we compare the airspace coverage of the three prism-shaped coverage structures.

\begin{proposition} 
	
	Let $\zeta_{1},\zeta_{2},\zeta_{3}$ denote the inter-region overlap ratio
	of the TP, SP, and HP coverage structures,
	respectively, with guaranteed seamless coverage. The inequality $\zeta_{1}<\zeta_{2}<\zeta_{3}$
	holds.
	
	\label{proposition: coverage overlap minimize}
	
\end{proposition}

\begin{proof}
	\label{proof2}
	When a TBS is equipped with directional antennas, its coverage
	takes on an approximately cone-shaped space. The coverage radius of
	the TBS is denoted as $r$, while $d_{1}$, $d_{2}$, and $d_{3}$ denote
	the  inter-site distance between TBSs for the TP, SP, and HP coverage structures, respectively.
	Considering the macro TBS, the height of the low-altitude space
	is denoted as $H$, where $0<H<r$. The inter-region coverage overlap ratio is defined as the ratio of the TBS signal coverage volume exceeding the coverage structure's range to the total volume.
	
	\begin{figure}[t]
		\centering
		\includegraphics[width=0.4\textwidth]{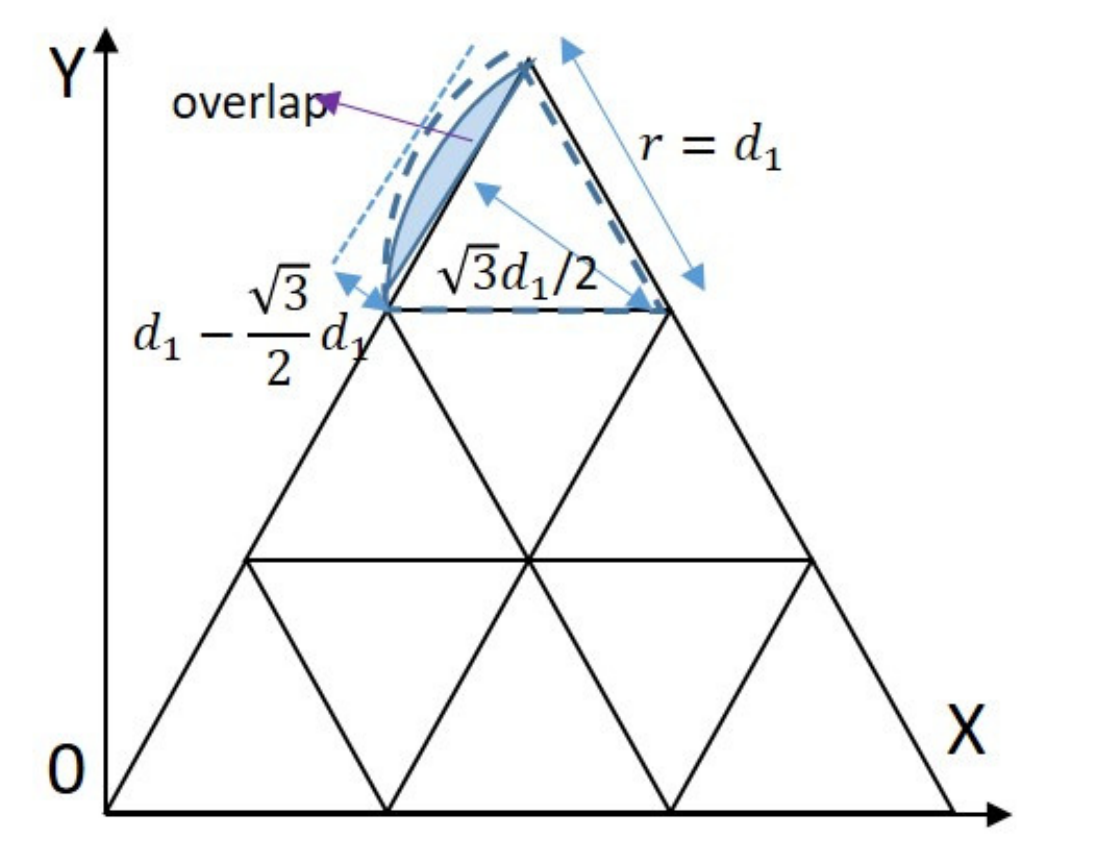}
		\caption{Illustration of the TP coverage structure. }
		\label{fig:Illustration-of-triangle-shaped}
	\end{figure}
	
	The vertical view of the TP structure is shown in Fig. \ref{fig:Illustration-of-triangle-shaped}.
	The area of the regular triangular region with side length $d_{1}=r$
	is given by $s_{1}=\sqrt{3}d_{1}^{2}/4$. The volume of the TP airspace is
	given by $V_{1}=s_{1}H.$ The blue arched region represents the inter-region
	overlap. As it rotates in 3D space, the arched region becomes a spherical
	segment. The volume of a spherical segment is given by $V_{11}=\pi(d_{1}-\sqrt{3}d_{1}/2)^{2}(r-(d_{1}-\sqrt{3}d_{1}/2)/3)$. Then, $\zeta_{1}$ is given by
	
	\begin{align}
		\zeta_{1}= & \frac{3V_{11}}{V_{1}}=\frac{(16\sqrt{3}-27)\pi r}{6H}\thickapprox0.37r/H.
	\end{align}
	
	The vertical view of the SP structure is shown in Fig. \ref{fig:Illustration-of-quadrangular-sha}.
	The area of the regular quadrangular with side length $d_{2}=\sqrt{2}r/2$
	is given by $s_{2}=d_{2}^{2}$. The volume of the SP airspace is
	given by $V_{2}=s_{2}H.$ The
	blue arched region contains the overlapped and non-overlapped parts. The volume of the spherical segment is given by $V_{21}=\pi(\sqrt{2}d_{2}-d_{2})^{2}(r-(\sqrt{2}d_{2}-d_{2})/3)$. In the blue region, the red triangular belongs to the non-overlapped part, forming a cone in the 3D airspace. The volume of the cone is given by $V_{22}=\pi((\sqrt{2}d_{2}-d_{2})^{2}(\sqrt{2}d_{2}-d_{2})/3)$. Then, $\zeta_{2}$ is given by

	\begin{figure}[t]
		\centering
		\includegraphics[width=0.41\textwidth]{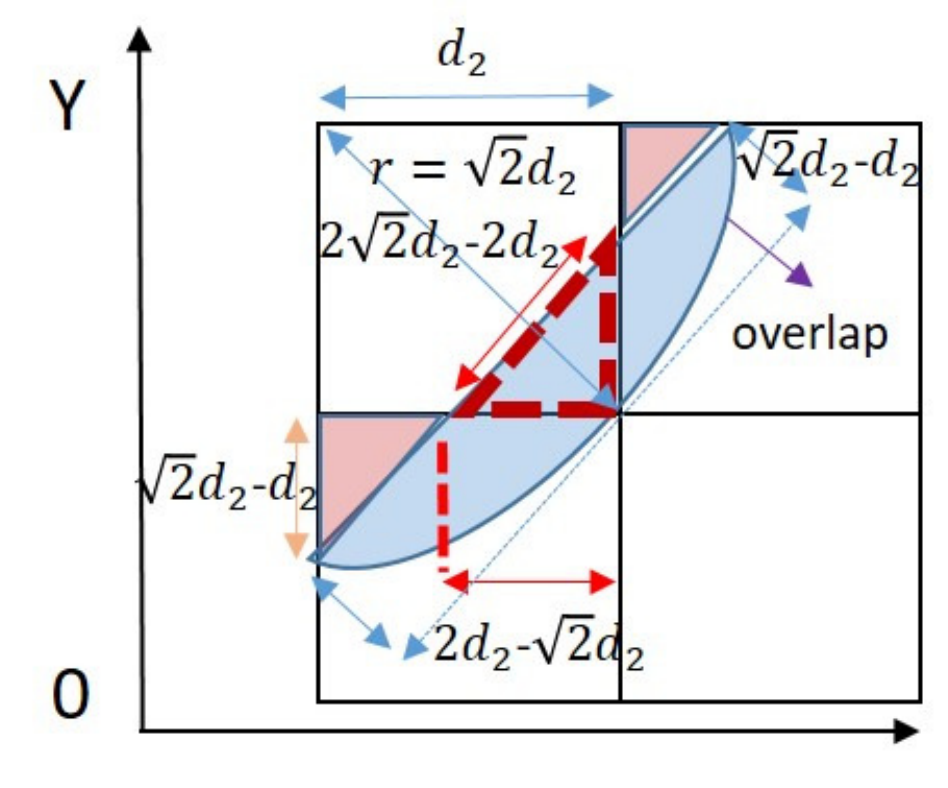}
		\caption{Illustration of the SP coverage structure. }
		\label{fig:Illustration-of-quadrangular-sha}
	\end{figure}
	
	\begin{align}
		\zeta_{2}\!= & \frac{4(V_{21}\!-\!V_{22})}{V_{2}}\!=\!\frac{4(\sqrt{2}\!-\!1)\pi r}{3H}\!\thickapprox\!1.74r/H.
	\end{align}
	
	The vertical view of the HP structure is shown in Fig. \ref{fig:Illustration-of-hexagonal-shaped}.
	The area of the regular hexagonal with side length $d_{3}=r/2$,
	is given by $s_{3}=3\sqrt{3}d_{3}^{2}/2$. The volume of the HP airspace is
	given by $V_{3}=s_{3}H.$ The
	blue arched region contains the overlapped and non-overlapped parts. The volume of the spherical segment is given by  $V_{31}=\pi(d_{3})^{2}(r-d_{3}/3)$. In the blue region, the red triangular and green rectangular belong to the non-overlapped parts, forming a cone  and a cylinder in the 3D space, respectively. The volume of the red cone is given by $V_{32}=\pi((\sqrt{3}d_{3}/2)^{2}(d_{3}/2))$. The
	volume of the green cylinder is given by $V_{33}=\pi((\sqrt{3}d_{3}/2)^{2}(d_{3}/6))$. Then, $\zeta_{3}$ is given by

	\begin{figure}[t]
		\centering
		\includegraphics[width=0.41\textwidth]{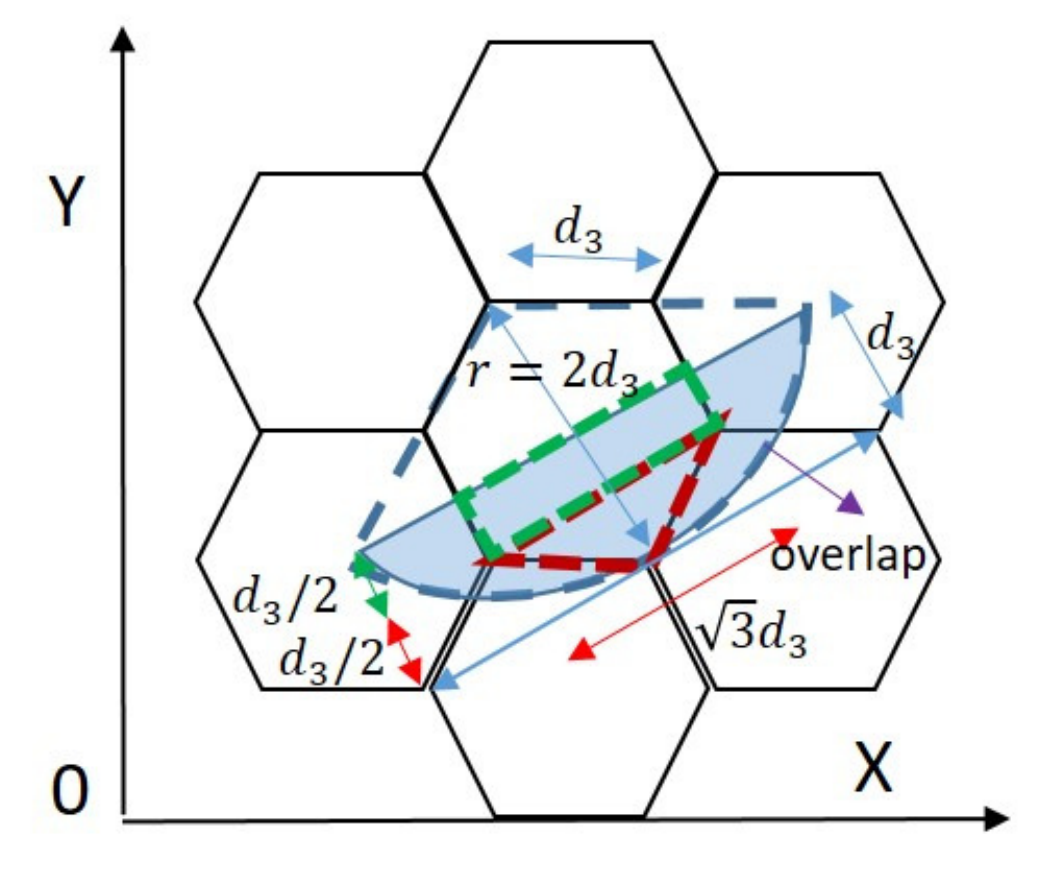}
		\caption{Illustration of the HP coverage structure. }
		\label{fig:Illustration-of-hexagonal-shaped}
	\end{figure}
	
	\begin{align}
		\zeta_{3}\!= & \frac{6(V_{31}\!-\!V_{32}\!-\!V_{33})}{V_{3}}\!=\!\frac{7\pi r^{3}}{3\sqrt{3}r^{2}H}\!\thickapprox\!4.23r/H.
	\end{align}
	
	Hence, we can obtain the inequality $\zeta_{1}<\zeta_{2}<\zeta_{3}$. 
\end{proof}

The conclusion drawn from Proposition \ref{proposition: coverage overlap minimize} is that the cooperation between three TBSs can obtain a minimum coverage overlap while ensuring seamless coverage.
	Therefore,
	we propose the CoTP model, which is based on cooperation among
	three TBSs.
As shown in Fig. \ref{fig:CoTP-Coverage-model}, all TBSs deployed
on plane $\mathrm{P}$ are clustered into multiple TBS sets, and each set includes three TBSs, such as $\{\textrm{BS}_{1},\textrm{BS}_{2},\textrm{BS}_{3}\}$.
The triangle-shaped region $\mathrm{P}_{1}$ is created by connecting
each pair of BSs in $\{\textrm{BS}_{1},\textrm{BS}_{2},\textrm{BS}_{3}\}$.
Correspondingly, the given region $\mathsf{P}\mathbb{\in R^{\mathbf{\textrm{2}}}}$
is divided into multiple triangular-shaped regions. A TP-shaped airspace $\mathrm{\varOmega}_{1}$ is formed by extending
$\mathrm{P}_{1}$ along the vertical direction. Then, the airspace
$\varOmega_{1}$ is divided into three subspaces, including the low-layer,
the medium-layer, and the high-layer, denoted as $\{\textrm{L}_{1},\textrm{L}_{2},\textrm{L}_{3}\}$
and served by the antennas of $\{\textrm{BS}_{1},\textrm{BS}_{2},\textrm{BS}_{3}\}$
with different parameters $\{\epsilon_{1},\epsilon_{2},\epsilon_{3}\}$,
respectively. 

\begin{figure}[t]
	\centering
	\includegraphics[width=0.41\textwidth]{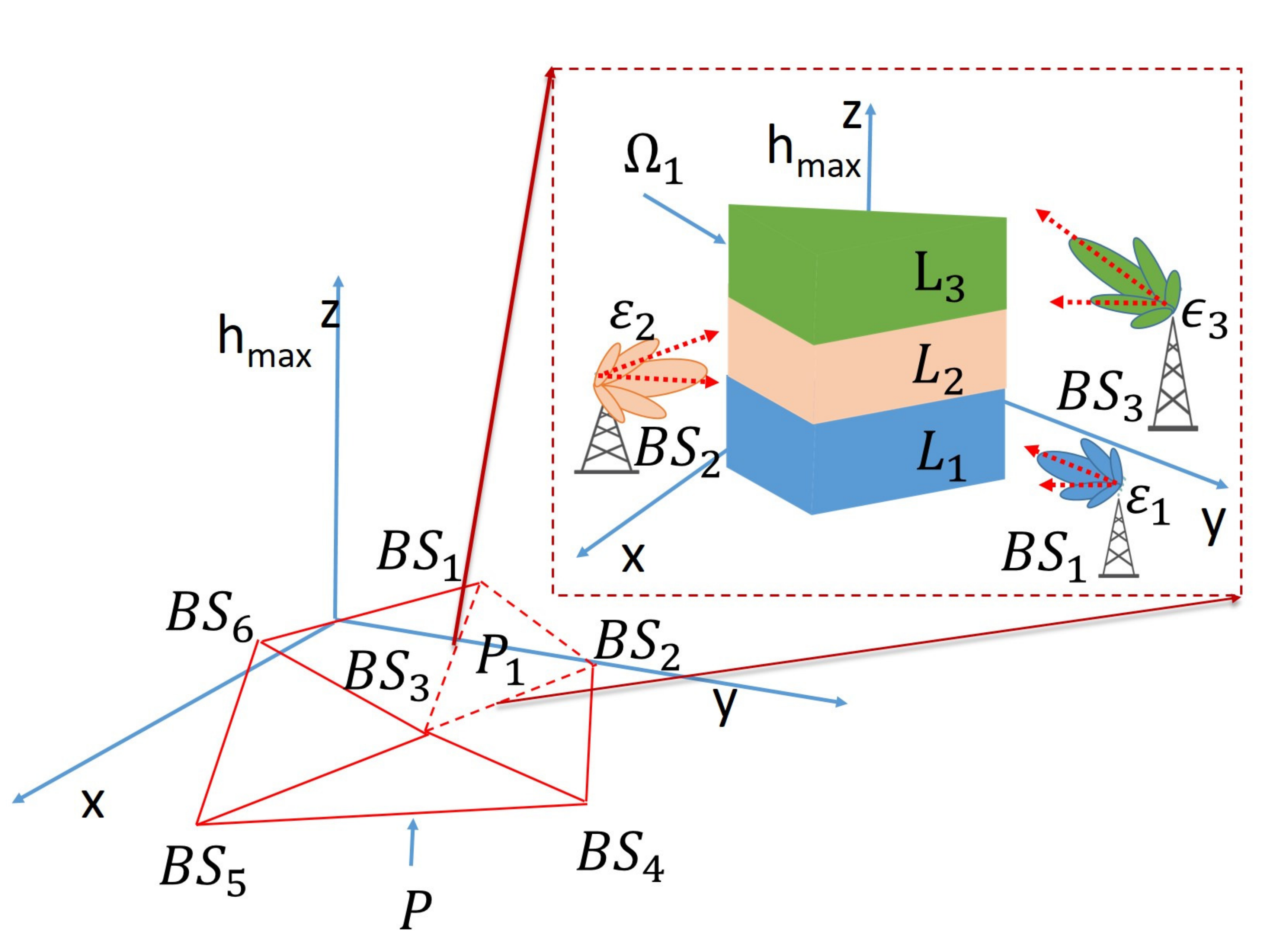}
	\caption{Illustration of the CoTP model.}
	\label{fig:CoTP-Coverage-model}
\end{figure}

\subsection{Coverage Structure Based on Delaunay Triangulation}

Based on the CoTP model, we design the Delaunay triangulation (DT)-based
cooperative coverage structure to enhance G2A coverage extension. The DT satisfies
both the empty circumcircle (EC) criterion and the maximized minimum
internal angle (MIAM) criterion \cite{Okabe-tessellations}, and its
excellent coverage ability has been extensively verified \cite{DT-optimization-Abedi}.
Furthermore, the DT-based coverage structure provides an advantage
in controlling intra-region overlap.

\begin{proposition} 
	
	If the antenna radiation radius of three TBSs located at the vertices
	of a triangle are equal, and the seamless coverage of a triangle-shaped
	region is guaranteed, then the ratio of intra-region overlap decreases
	as the length of the longest edge of the triangle decreases. 
	
	\label{proposition: overlap minimize}
	
\end{proposition}

\begin{proof}
	\label{proof1}
	Let $S_{0}$, and $L$ denote the area and the length
	of the longest edge of the triangle, respectively. For a given triangle-shaped region,
	$S_{0}$ remains constant.
	
	To ensure seamless coverage of the triangle-shaped region, the transmission power of the antenna is determined
	by the length of the longest edge $L$. Because the radiation radius of three TBSs located at the vertices of a triangle are the same, this radius can be denoted as $R=0.5L$. Then, the total coverage area of all corresponding antennas can be calculated
	as
	
	\begin{equation}
		S=\sum_{i=1,2,3}\frac{n_{i}\pi R^{2}}{360},
	\end{equation}
	where $n_{1}$, $n_{2}$, and $n_{3}$ are the inner angles of the
	triangle. Due to $n_{1}+n_{2}+n_{3}=180^{\text{\textdegree}}$,
	we have $S=0.5\pi R^{2}$. As $L$ decreases, $R$ and $S$ also decrease. 
	
	Therefore, we can conclude that the coverage overlap ratio, which is given as $\zeta=(S-S_{0})/S_{0}$, also decreases as the length of the longest edge of the triangle decreases.
\end{proof}

The DT can minimize the length of the longest edge of the triangle
statistically \cite{Okabe-tessellations}. Based on the conclusion
of Proposition \ref{proposition: overlap minimize}, the DT-based
coverage structure can minimize the intra-region overlap statistically.
With the DT-based coverage structure, all TBSs in a cellular network are
clustered into multiple TBS sets, and each set contains three TBSs. The
plane enclosed by three TBSs in each set can be vertically extended
into a TP-shaped airspace.


\subsection{Space Layered Beam Cooperative Coverage Algorithm}

In the paper \cite{Robust-Planning}, the authors proposed a randomized optimization method designed for clustered networks. Drawing inspiration from this, we have devised a Particle Swarm Optimization (PSO)-based algorithm specifically tailored for each clustered network. In the TP-shaped airspace, the PSO-based SLBC coverage algorithm for obtaining optimal beam parameters is proposed. The SLBC algorithm can tackle the challenge arising from the coexistence of continuous and discrete variables. As shown in Fig. \ref{fig:SLBS-coverage},
three beams of $\{\textrm{BS}_{1},\textrm{BS}_{2},\textrm{BS}_{3}\}$
implement cooperative airspace coverage. For the existing antenna, the H-HPBW
and V-HPBW are discrete while the tilt angle is continuous \cite{Francisco-5G-pattern}.
Especially, the beam pattern determines the H-HPBW and V-HPBW. In the SLBC, two groups of particle swarms work separately.
The discrete one $\varXi_{1}$ corresponds to the beam pattern while
the continuous one $\varXi_{2}$ corresponds to the tilt angle.
Both groups search for solutions independently and share the same
fitness function (\ref{eq:objetive-function}).

\begin{figure}[t]
	\centering
	\includegraphics[width=0.43\textwidth]{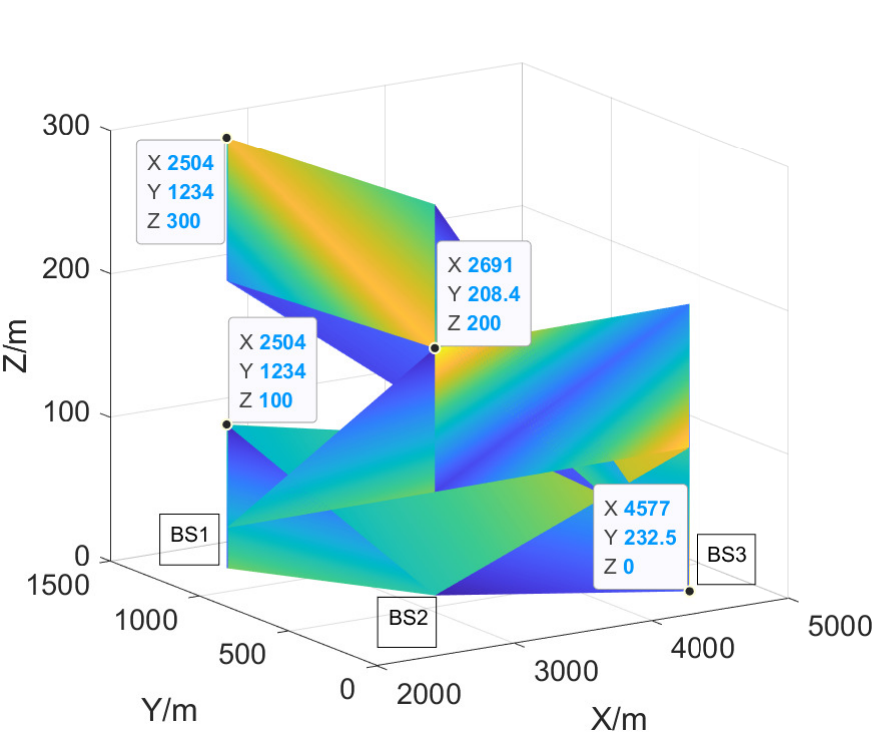}
	\caption{The diagram of cooperative beam coverage. }
	\label{fig:SLBS-coverage}
\end{figure}

To explore the solution, $N$ particles
search in the $M$-dimensional solution space. During the $l-$th
iteration, the tilt angle optimization task in swarm $\varXi_{1}$
involves updating the velocity and position of each particle according
to (\ref{eq:CPSO update}), 

\begin{equation}
	\begin{cases}
		V_{m}(l\!\!+\!\!1)\!\!=\!\!wV_{m}(l)\!\!+\!\!c_{1}F_{1}(S_{L}^{*}(l)\!\!-\!\!B_{m}(l))\\
		\ \ \ \ \ \ \ \ \ \ \ \ \ \ +c_{2}F_{2}(S_{G}^{*}(l)-B_{m}(l)),\\
		B_{m}(l\!\!+\!\!1)\!\!=\!\!B_{m}(l)\!\!+\!\!V_{m}(l\!\!+1),m\in[1,M],
	\end{cases}\label{eq:CPSO update}
\end{equation}
where $w$ denotes the inertia weight, $S_{L}^{*}(l)$ denotes the
local best position, and $S_{G}^{*}(l)$ denotes the global best position. Moreover, $F_{1}$ and $F_{2}$ are a pair of random coefficients. $c_{1}$
and $c_{2}$ denote the local and global learning coefficients, respectively. The learning coefficient controls the step size of
	particle position updates based on individual experience and group
	cooperation. The inertia weight reflects the tendency of particles
	to maintain their current velocity. Generally, the inertia weight
	$w$ gradually decreases during iterations to encourage particles
	to converge throughout the search process, as given by $w=w_{\mathrm{max}}-(l-1)\times(w_{\mathrm{max}}-w_{\mathrm{min}})/(N_{\mathrm{iter}}-1)$,
	where $N_{\mathrm{iter}}$ denotes the total number of iterations,
	$w_{\mathrm{max}}$ and $w_{\mathrm{min}}$ denote the bounds of inertia
	weight.

In addition, each particle in swarm $\varXi_{2}$ performing the beam
pattern optimization task requires a discrete version of iteration
to update their velocity and position, following equation (\ref{eq:DPSO update}),

\begin{equation}
	\begin{cases}
		G_{m}(l\!\!+\!\!1)\!\!=\!\!gG_{m}(l)\!\!+\!\!d_{1}F_{1}(S_{L}^{*}(l)\!\!-\!\!X_{m}(l))\\
		\ \ \ \ \ \ \ \ \ \ \ \ \  \ +d_{2}F_{2}(S_{G}^{*}(l)-X_{m}(l)),\\
		X_{m}(l\!\!+\!\!1)\!\!=\!\!\left\lceil X_{m}(l)\!\!+\!\!G_{m}(l\!\!+\!\!1)\right\rceil,m\!\!\in\!\![1,\!\!M],
	\end{cases}\label{eq:DPSO update}
\end{equation}
where $\left\lceil \;\;\right\rceil $ denotes the operation of rounding
down. Different from the general PSO, the iteration is performed
under the constraint of overlap ratio, as described in (\ref{eq:overlap-constraints}). Each of two particle swarms has the parameter $M=3$, corresponding
to three TBSs. The SLBC
algorithm is summarized in Algorithm \ref{alg:HPSO}.

\begin{algorithm}[t]
	\caption{SLBC coverage algorithm \label{alg:HPSO}}
	
	\small  
	
	\textbf{Input: }Initialize the learning coefficients, the inertia weight of particle swarms.
	
	\textbf{Output:} The optimal GCR.
	
	\begin{algorithmic}[1]
		
		\STATE \textbf{Initialization}
		
		\STATE Randomly initialize the positions, velocity of $N$ particle
		swarms in $\varXi_{1}$ and $\varXi_{2}$.
		
		\STATE Calculate the global best fitness value of all particles according
		to (\ref{eq:objetive-function}).
		
		\STATE Calculate the global best position of particles in $\varXi_{1}$
		and $\varXi_{2}$.
		
		\FOR{$i$ in range $N_{iter}$}
		
		\FOR{$j$ in range $N$}
		
		\STATE Update the velocity and position of $\varXi_{1}$ as
		(\ref{eq:CPSO update}).
		
		\STATE Update the velocity and position of $\varXi_{2}$ as
		(\ref{eq:DPSO update}).
		
		\STATE Calculate the GCR and COR according to (\ref{eq:objetive-function}).
		
		\IF{COR satisfies the constraint (\ref{eq:overlap-constraints})}
		
		\STATE Update local best fitness value of particle swarms.
		
		\STATE Update local best position of particle swarms.
		
		\ENDIF\
		
		\ENDFOR\
		
		\STATE Update the global best fitness value of particle swarms.
		
		\STATE Update the global best position of all particle swarms.
		
		\ENDFOR\
		
		\RETURN The optimal GCR.
		
	\end{algorithmic}
\end{algorithm}

%
%
%
%
%

\subsection{Adaptive Beam Cooperative Coverage Algorithm}

To track the challenge posed by the limited number of beam patterns,
a novel algorithm is required to improve the coverage of arbitrary triangular-shaped airspace. Additionally, it is crucial to prevent power leakage
into adjacent TP-shaped spaces. Thus, the ABC
algorithm is designed to determine the optimal and continuous H-HPBW, V-HPBW, and tilt angle.

Considering that the ABC algorithm originates from the PSO, the particles
search for the solution in $M$-dimensional solution space $S=\{\varOmega_{1},\:\varOmega_{2},\:\varOmega_{3},\:V_{1},\:V_{2},\:V_{3},\:\theta_{1},\:\theta_{2},\:\theta_{3}\}$,
where $\varOmega_{1},\varOmega_{2},\varOmega_{3}$ denote continuous H-HPBW, $V_{1},V_{2},V_{3}$ denote continuous V-HPBW, and $\theta_{1},\theta_{2},\theta_{3}$
denote the continuous tilt angle of three TBSs, respectively. In addition, the flow of the ABC
algorithm is similar to that of the general PSO algorithm.

\section{NUMERICAL RESULTS\label{sec:NUMERICAL-RESULT}}

In this section, we present simulations and field trials to evaluate
the performance of the CoTP-based G2A coverage extension method.

\subsection{Simulation Results}

In the simulations, the environment parameters are set as $h_{max}=300\textrm{m}$,
$\tau=-90\textrm{dBm}$, and $P_{\textrm{T}}=46\textrm{dBm}$ \cite{TR36.777}. The learning coefficients and inertia weights are set as $c_{1}=d_{1}=1.5$, $c_{2}=d_{2}=2.5$, $w_{\mathrm{min}}=0.4$
	, and $w_{\mathrm{max}}=0.9$. $a_{1},a_{2},$ and $a_{3}$ denote three inner angles
of a given triangle. The coverage ratio obtained in a TP airspace is treated as the optimal
GCR. $v_{1}$ and $v_{2}$ are used to denote the optimal GCR based
on the SLBC algorithm and ABC algorithm, respectively. In order to
evaluate the G2A coverage extension in large-scale networks, the average GCR is defined as $\varrho=\sum_{t=1}^{\mathsf{D}}b_{t}s_{t}/\sum b_{t},$
where $b_{t}$ and $s_{t}$ denote the area of each triangle-shaped region and the optimal GCR of each TP-shaped airspace.  D represents the number of DT-based triangular
regions.

%

We simulate the GCR of a synthetic network that follows
the configuration of a real-world operator's network, as shown in Fig.
\ref{fig:Test-Environment}. The simulation area comprises 9 TBSs,
covering a 15 square kilometer region. The results are depicted in
Table \ref{tab:The-Coverage-Ratio-1}. When the COR is 0.01\%, the average
GCR with the SLBC algorithm can reach 86\%, while the average GCR
with the ABC algorithm is over 93\%.
The average GCR of the ABC algorithm is higher than that of the SLBC
algorithm. The advantage of the
ABC algorithm is derived from its continuous adjustment of the H-HPBW and V-HPBW parameters.

\begin{table}[t]
	\begin{centering}
		\caption{The performance of the CoTP-based coverage extension method.\label{tab:The-Coverage-Ratio-1}}
		\par\end{centering}
	\centering{}%
	\begin{tabular}{|c|c|c|c|c|c|c|}
		\hline 
		TBS set & $\alpha_{1}$ & $\alpha_{2}$ & $\alpha_{3}$ & $b$ & $v1$ & $v2$\tabularnewline
		\hline 
		\hline 
		$\mathrm{X_{1}X_{2}X_{7}}$ & 5 & 171 & 5 & 0.61 & 0.51 & 0.84\tabularnewline
		\hline 
		$\mathrm{X_{1}X_{2}X_{4}}$ & 140 & 15 & 25 & 0.67 & 0.85 & 0.95\tabularnewline
		\hline 
		$\mathrm{X_{1}X_{4}X_{5}}$ & 25 & 121 & 34 & 0.66 & 0.80 & 0.98\tabularnewline
		\hline 
		$\mathrm{X_{1}X_{5}X_{6}}$ & 34 & 35 & 111 & 2.22 & 0.88 & 0.95\tabularnewline
		\hline 
		$\mathrm{X_{2}X_{3}X_{4}}$ & 27 & 54 & 99 & 0.97 & 0.9 & 0.95\tabularnewline
		\hline 
		$\mathrm{X_{2}X_{3}X_{7}}$ & 23 & 82 & 75 & 1.36 & 0.91 & 0.94\tabularnewline
		\hline 
		$\mathrm{X_{3}X_{4}X_{5}}$ & 70 & 74 & 36 & 1.56 & 0.87 & 0.94\tabularnewline
		\hline 
		$\mathrm{X_{5}X_{6}X_{9}}$ & 65 & 75 & 40 & 4.62 & 0.87 & 0.92\tabularnewline
		\hline 
		$\mathrm{X_{6}X_{8}X_{9}}$ & 61 & 66 & 53 & 2.5 & 0.9 & 0.93\tabularnewline
		\hline 
	\end{tabular}
\end{table}

The simulation results using a random plane division method are presented
in Table \ref{tab:The-Coverage-Ratio-2}. Comparing the results
in Table \ref{tab:The-Coverage-Ratio-1} and Table \ref{tab:The-Coverage-Ratio-2},
it is observed that the average GCR of the DT-based coverage structure
is 86\%, whereas the coverage structure based on the random plane division
achieves an average GCR of 78\% with the SLBC algorithm. Therefore,
it can be concluded that the DT-based coverage structure is better
than the coverage structure based on the random plane division
method. With the ABC algorithm, the conclusion still holds. 

\begin{table}[t]
	\begin{centering}
		\caption{The performance of coverage extension method based on random division.\label{tab:The-Coverage-Ratio-2}}
		\par\end{centering}
	\centering{}%
	\begin{tabular}{|c|c|c|c|c|c|c|}
		\hline 
		TBS set & $\alpha_{1}$ & $\alpha_{2}$ & $\alpha_{3}$ & $b$ & $v1$ & $v2$\tabularnewline
		\hline 
		\hline 
		$\mathrm{X_{3}X_{5}X_{7}}$ & 172 & 4 & 4 & 0.43 & 0.55 & 0.88\tabularnewline
		\hline 
		$\mathrm{X_{2}X_{5}X_{7}}$ & 141 & 20 & 19 & 2.14 & 0.81 & 0.94\tabularnewline
		\hline 
		$\mathrm{X_{1}X_{2}X_{7}}$ & 5 & 171 & 5 & 0.61 & 0.51 & 0.84\tabularnewline
		\hline 
		$\mathrm{X_{2}X_{4}X_{5}}$ & 33 & 100 & 47 & 1.31 & 0.90 & 0.94\tabularnewline
		\hline 
		$\mathrm{X_{1}X_{2}X_{4}}$ & 140 & 15 & 25 & 0.67 & 0.85 & 0.95\tabularnewline
		\hline 
		$\mathrm{X_{1}X_{4}X_{5}}$ & 25 & 121 & 34 & 0.66 & 0.80 & 0.98\tabularnewline
		\hline 
		$\mathrm{X_{5}X_{8}X_{9}}$ & 19 & 35 & 126 & 3.47 & 0.80 & 0.86\tabularnewline
		\hline 
		$\mathrm{X_{1}X_{5}X_{8}}$ & 101 & 55 & 24 & 4.89 & 0.79 & 0.87\tabularnewline
		\hline 
		$\mathrm{X_{1}X_{6}X_{8}}$ & 11 & 159 & 10 & 0.99 & 0.75 & 0.89\tabularnewline
		\hline 
	\end{tabular}
\end{table}

To illustrate the performance of the SLBC algorithm and ABC algorithm,
a simulation is conducted in the triangle-shaped region enclosed by
${\mathrm{X_{2},X_{3},X_{4}}}$, where three inner angles are $a_{1}=99^{\text{\textdegree}},\ensuremath{a_{2}}=5\ensuremath{4^{\text{\textdegree}}},a_{3}=27^{\text{\textdegree}}$. 

Fig. \ref{fig:SLBC-CE} depicts the performance of the SLBC algorithm. The simulation results demonstrate that the SLBC algorithm can rapidly
converge, and the optimal GCR can reach 90\% while the COR is 0.01\%.
Moreover, the optimal tilt angle, H-HPBW, V-HPBW for three TBSs
located in the TP-shaped airspace are $\{-52^{\circ},15^{\circ},25^{\circ}\}$,
$\{-13^{\circ},110{}^{\circ},25{}^{\circ}\}$, $\{9^{\circ},110^{\circ},25^{\circ}\}$,
respectively. The tilt angles indicate the subspace in the TP-shaped airspace
that each antenna is responsible for. The optimal parameters obtained in the simulation provide valuable guidance for G2A coverage extension. However, it is worth
noting that the beam power may leak to adjacent TP-shaped airspace
since the H-HBPW value of $110^{\circ}$ exceeds $a_{2}$ and
$a_{3}$. The result means that the width of the beam is wider than
the given region. 

\begin{figure}[t]
	\centering
	\includegraphics[width=0.42\textwidth]{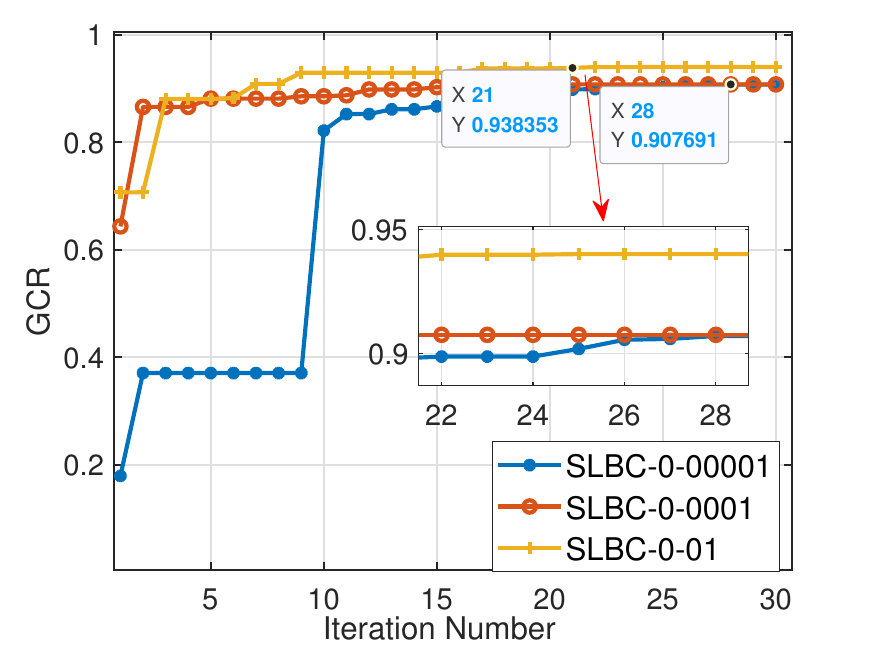}
	\caption{The convergence of the SLBC algorithm. SLBC-0.00001, SLBC-0.0001, and SLBC-0.01 represent the cases where the COR values are 0.00001, 0.0001, and 0.01, respectively.}
	\label{fig:SLBC-CE}
\end{figure}

The simulation results in Fig. \ref{fig:SB-CE} demonstrate that the ABC algorithm can rapidly
converge, and the optimal GCR can reach 95\% while the COR is 0.01\%.
Moreover, the optimal tilt angle, H-HPBW, V-HPBW of three TBSs
are $\{-80^{\circ},77^{\circ},112^{\circ}\}$, $\{-21^{\circ},35{}^{\circ},41{}^{\circ}\}$,
$\{14^{\circ},27^{\circ},129^{\circ}\}$, respectively. Note
that the H-HBPW values of all TBSs are smaller than the width of the given
region. This ensures that the main lobe of the antenna beam
falls within the triangular coverage region and avoids power leakage
to adjacent TP airspace.

\begin{figure}[t]
	\centering
	\includegraphics[width=0.42\textwidth]{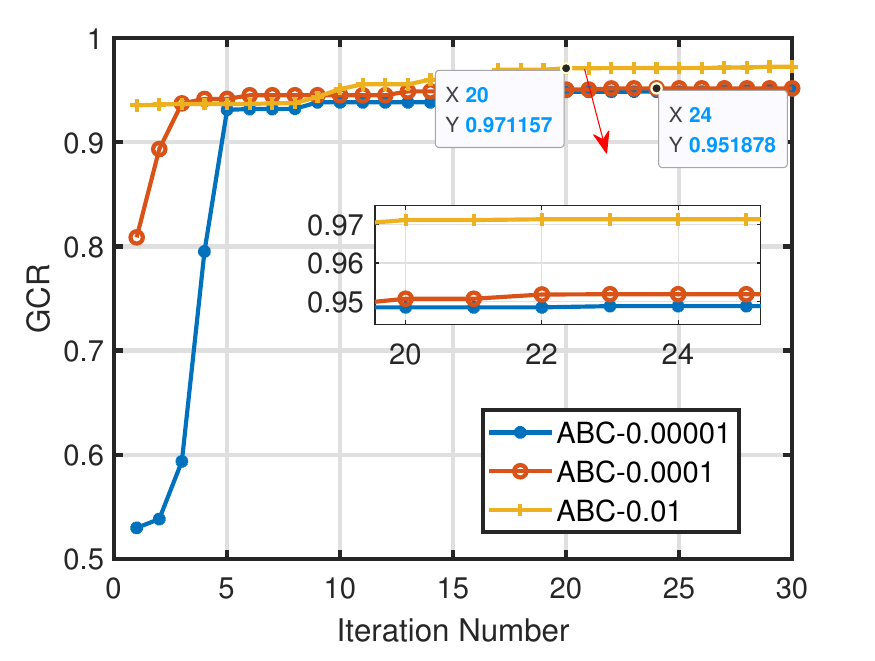}
	\caption{The convergence of the ABC algorithm. ABC-0.00001, ABC-0.0001, and ABC-0.01 represent the cases where the COR values are 0.00001, 0.0001, and 0.01, respectively.}
	\label{fig:SB-CE}
\end{figure}

Furthermore, we assess the GCR of the ABC and SLBC by comparing them to the exhaustive search (ES) algorithm and the algorithm that utilizes the 5G beam without 3D cooperation (5G Beam W/O) \cite{Francisco-5G-pattern}.

\begin{figure}[t]
	\centering
	\includegraphics[width=0.43\textwidth]{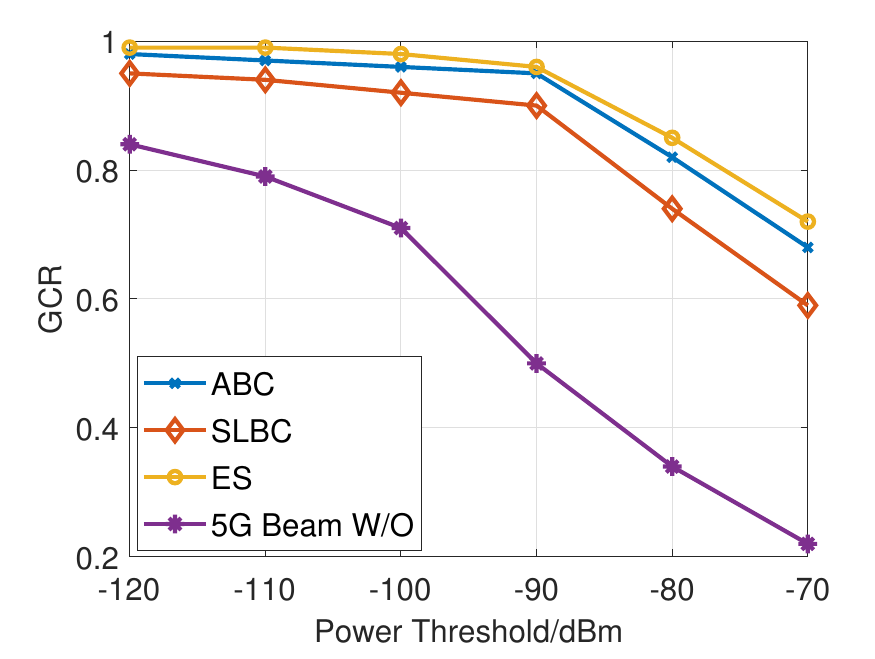}
	\caption{The coverage for airspace with varying received power thresholds.}
	\label{fig:Different-Threshold}
\end{figure}

The GCR under different received power thresholds is shown
	in Fig. \ref{fig:Different-Threshold}. It is observed that as the threshold increases, the GCR decreases. Furthermore, the GCR of the ABC closely aligns with that of the ES algorithm, exceeding
	the GCR of the 5G beam-based algorithm without cooperation.

\begin{figure}[t]
	\centering
	\includegraphics[width=0.43\textwidth]{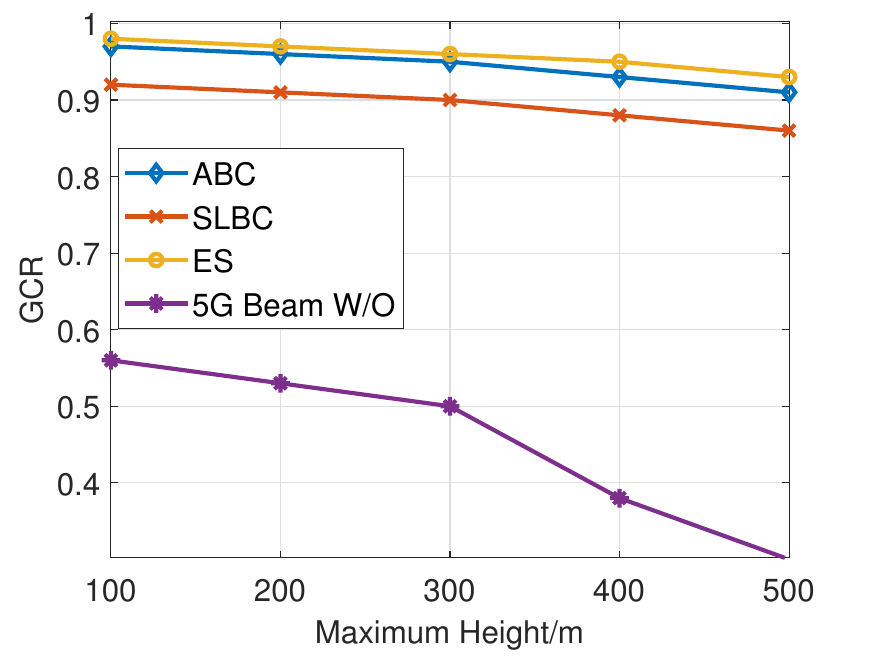}
	\caption{The coverage for airspace with varying maximum heights.}
	\label{fig:Different-Height}
\end{figure}

The GCR under varying maximum heights
	is shown in Fig. \ref{fig:Different-Height}. It should be noted that an increase in the maximum
	height of the covered airspace results in a slight decline in the GCR.
	Nonetheless, the proposed algorithm can maintain robust 3D coverage performance even as height increases. This suggests that the proposed algorithms obtain robust adaptability to variations
	in coverage height.



\subsection{Field Trials}

The field trials are conducted in a typical suburban region in Zigong,
a city located in southwest China. The networks include
9 TBSs, covering a region of 15 square kilometers. The network topology is shown in Fig. \ref{fig:Test-Environment},
with $\{\mathrm{X}_{1},\mathrm{X}_{2},\cdots,\mathrm{X}_{9}\}$ representing the
locations of the 9 TBSs. Multiple triangular-shaped regions have been
designed using the DT-based coverage structure generating method.


\begin{figure}[t]
	\centering
	\subfloat[\label{fig:Test-Environment}The TBSs deployment topology.]{\includegraphics[width=.5\linewidth]{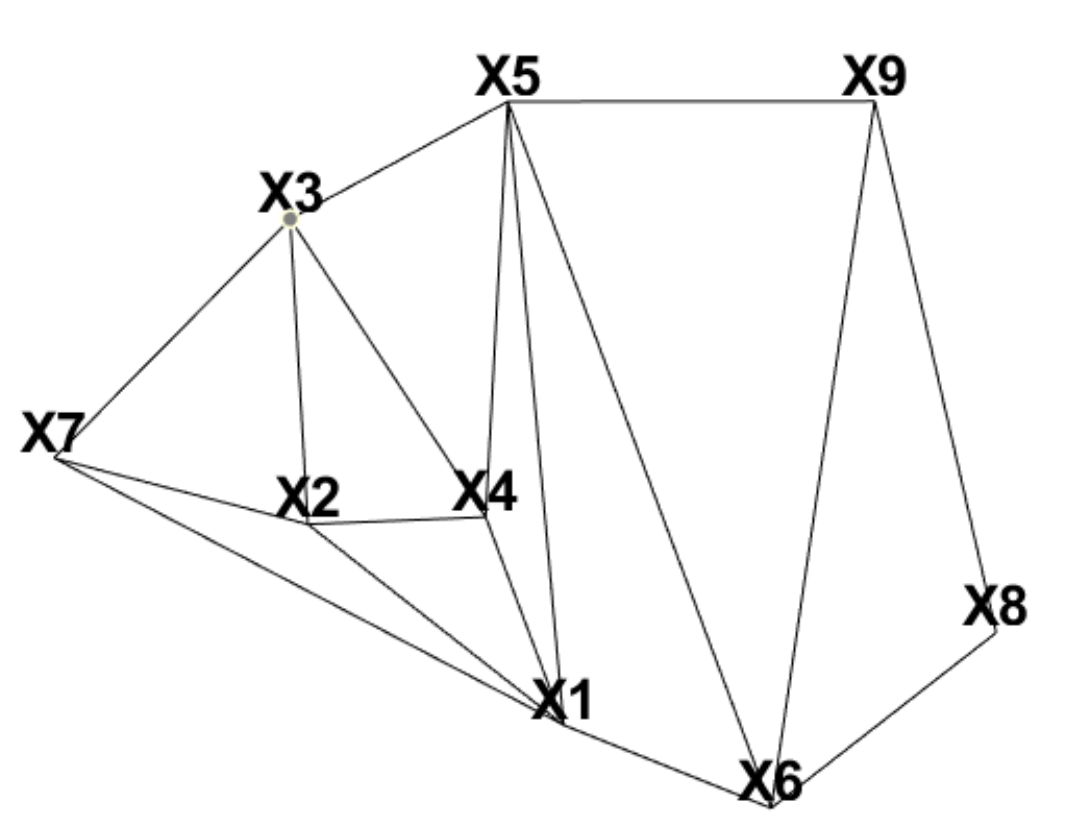}}
	\subfloat[\label{fig:UAV-and-drive-1}The drive-test equipment.]{\includegraphics[width=.5\linewidth]{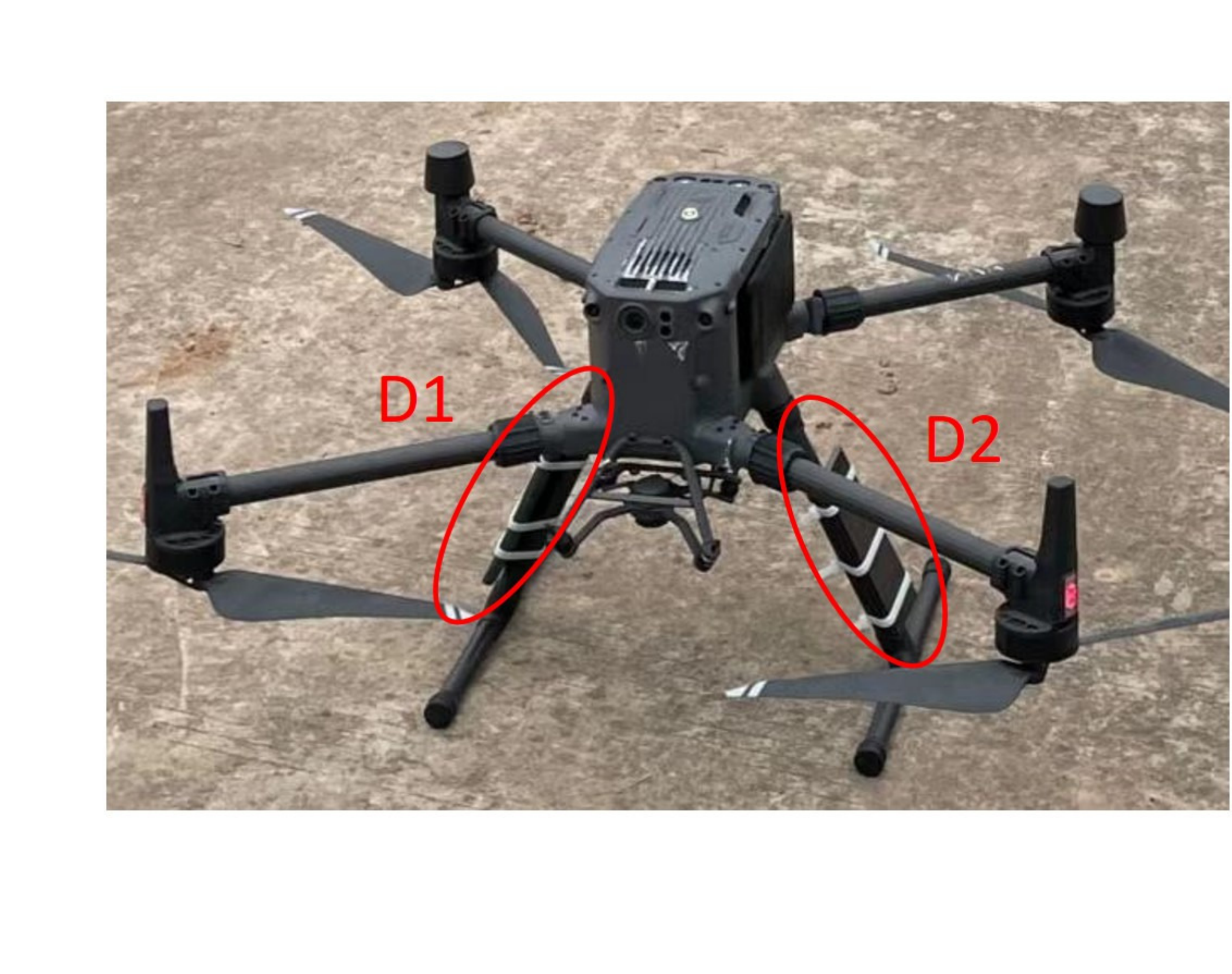}} \\
	\subfloat[\label{fig:UAV terminal}UAV control terminal.]{\includegraphics[width=.5\linewidth]{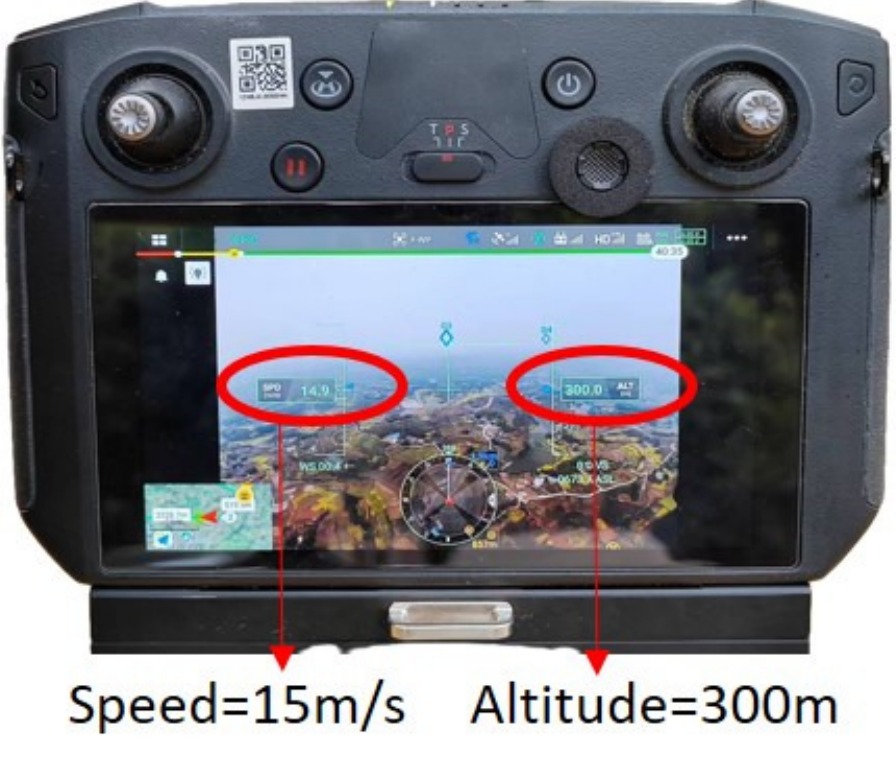}}
	\subfloat[\label{fig:UAV performing tasks}UAV performing tasks.]{\includegraphics[width=.5\linewidth]{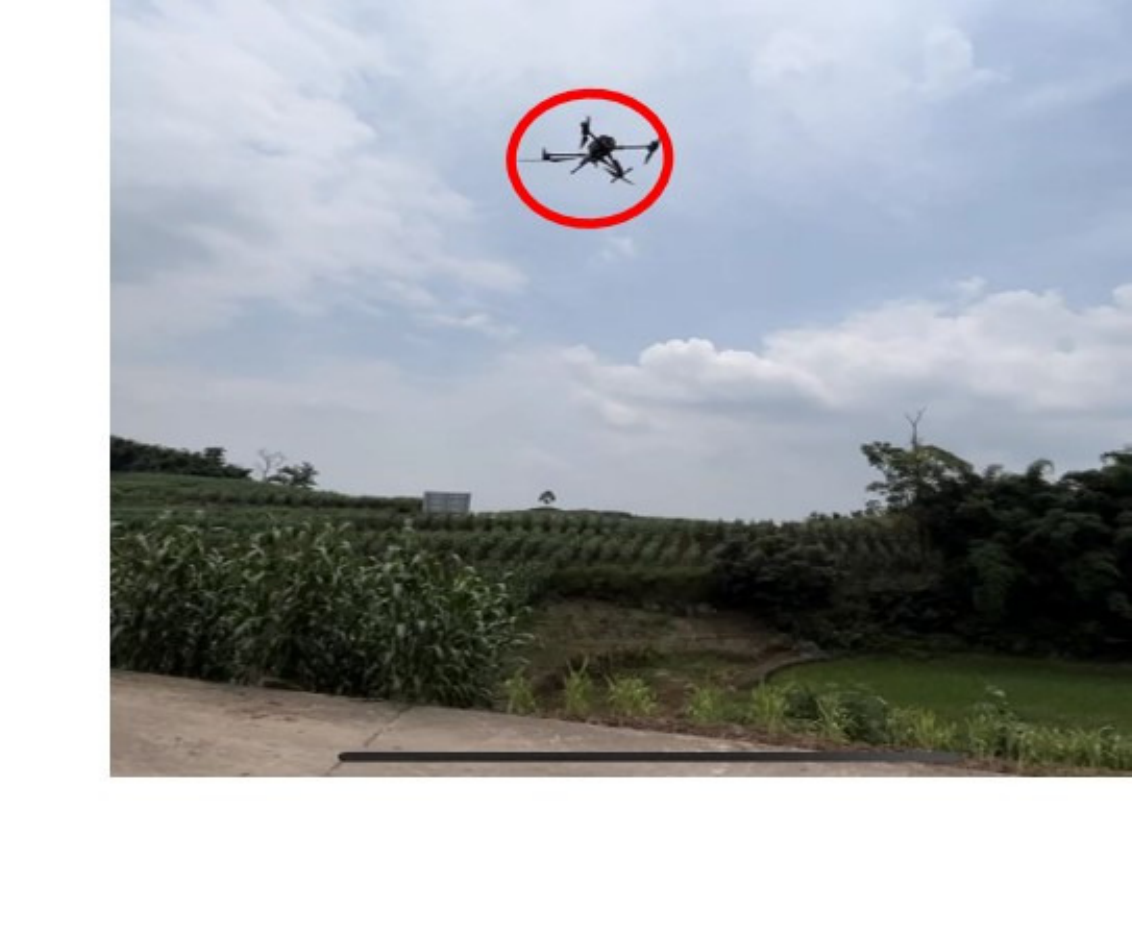}}	
	\caption{The field trial.}
	\label{fig:arms}
\end{figure}

	%
	%
	%
	%



The field trials are implemented by a UAV and
two drive-test terminals (see Fig. \ref{fig:UAV-and-drive-1}). The gathered samples from two drive-test terminals are labeled D1 and
D2, respectively.
On the drive-test terminals, the Huawei AAU5613 5G antennas are employed. The beam patterns are configured according to the 5G antenna mode \cite{Francisco-5G-pattern}. Due to operational limitations, we restrict the range of tilt
angle to $[-15^{\circ},15^{\circ}]$. The sample tasks are implemented at
different altitudes (Fig. \ref{fig:UAV terminal}, see Fig. \ref{fig:UAV performing tasks}). The detailed parameters
of the field trial are shown in Table \ref{tab:Parameter-of-Measurement}. 

\begin{table}[t]
	\caption{Parameters of the field trial.\label{tab:Parameter-of-Measurement}}
	
	\centering{}%
	\begin{tabular}{|c|c|}
		\hline 
		Parameters & Configuration\tabularnewline
		\hline 
		\hline 
		Propagation environment & Suburban\tabularnewline
		\hline 
		Frequency & 2.6GHz\tabularnewline
		\hline 
		Antenna & 64T$\times$64R\tabularnewline
		\hline 
		Number of aerial samples & 5000$\times$6\tabularnewline
		\hline 
		Speed of UAV & 15m/s\tabularnewline
		\hline 
		Flight altitude & \{75,100,150,200,250,300\}m\tabularnewline
		\hline 
		Sample rate & 1Hz\tabularnewline
		\hline 
	\end{tabular}
\end{table}

The coverage performance comparison between the down-tilted antennas and the
configurations based on the SLBC is shown in Table \ref{tab:GCR-of-test}.
We can observe that the network
with down-tilted antennas achieves an average GCR of around 46.05\%, while
the network with the antenna configurations based on the SLBC algorithm
obtains an average GCR of around 84.56\%. It is concluded that the GCR of the CoTP-based method improves by 83\% compared to a cellular network that uses down-tilted antennas.

\begin{table}[t]
	\begin{centering}
		\caption{The real-world measurement coverage performance comparison between the down-tilted antennas and the SLBC-based method. *-ori denotes the down-tilted antennas. *-opti denotes the
			configurations based on the SLBC-based method. D1 and D2 represent two different drive-test terminals. \label{tab:GCR-of-test}}
		\par\end{centering}
	\centering{}%
	\begin{tabular}{|c|c|c|c|c|}
		\hline 
		Parameters & D1-ori & D2-ori & D1-opti & D2-opti\tabularnewline
		\hline 
		\hline 
		$\mathrm{X}_{2}$-Beam (Pattern) & 2 & 2 & 7 & 7\tabularnewline
		\hline 
		$\mathrm{X}_{2}$-Tilt ($\text{\textdegree}$) & -3 & -3 & 3 & 3\tabularnewline
		\hline 
		$\mathrm{X}_{3}$-Beam (Pattern) & 3 & 3 & 8 & 8\tabularnewline
		\hline 
		$\mathrm{X}_{3}$-Tilt ($\text{\textdegree}$) & -3 & -3 & 7 & 7\tabularnewline
		\hline 
		$\mathrm{X}_{4}$-Beam (Pattern) & 4 & 4 & 9 & 9\tabularnewline
		\hline 
		$\mathrm{X}_{4}$-Tilt ($\text{\textdegree}$) & -1 & -1 & 11.5 & 11.5\tabularnewline
		\hline 
		GCR & 45.87\% & 46.23\% & 84.23\% & 84.89\%\tabularnewline
		\hline 
	\end{tabular}
\end{table}

Fig. \ref{fig:Different-sub-Layer-of} displays the GCR of sub-layer
airspace with different heights, presenting both the simulation results and field trials derived from the CoTP-based method and down-tilted antennas. It is observed that the CoTP-based
method enhances the coverage of each sub-layer of  the total airspace,
with a particular emphasis on ensuring coverage for near-ground
sub-layers, including (0-50)m and (50-100)m.

\begin{figure}[t]
	\centering
	\includegraphics[width=0.45\textwidth]{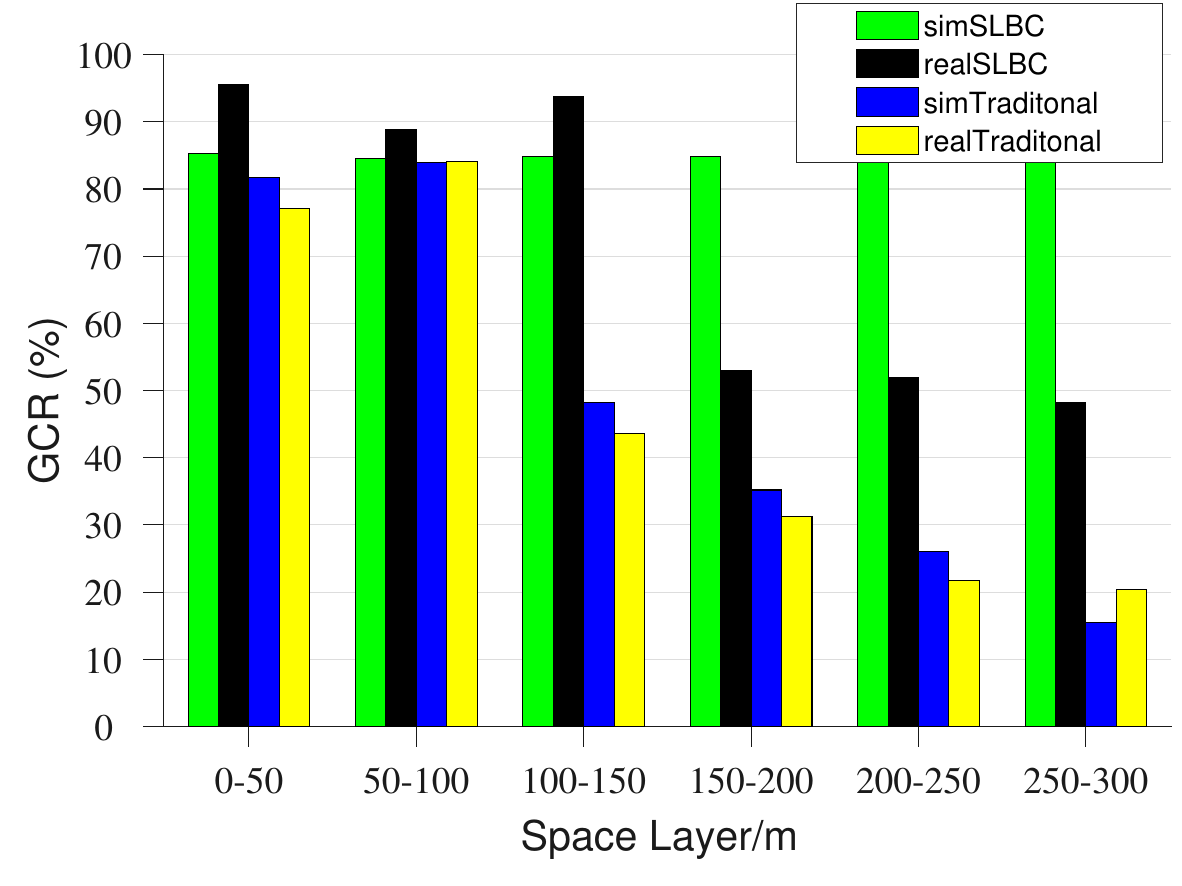}
	\caption{The coverage of varying subspace. The green bars represent the simulated GCR with the CoTP-based method, while the black bars represent the corresponding measurements in the real world. The blue bars represent the simulated GCR with down-tilted antennas, while the yellow bars represent the measurements with down-tilted antennas.}
	\label{fig:Different-sub-Layer-of}
\end{figure}

\section{CONCLUSION\label{sec:CONCLUSION}}

In this paper, we proposed a CoTP-based G2A coverage extension method
for beyond 5G networks. The proposed method aims to maximize the coverage
ratio of 3D airspace while controlling overlap, achieved by designing
the TBSs cooperative set and antenna parameters in existing cellular networks. First, we designed the
CoTP model and the coverage structure based on
Delaunay triangulation. Then, the SLBC algorithm has been designed to obtain
sub-optimal antenna parameters. To further improve the coverage ratio
in 3D airspace, the ABC algorithm has been designed, which can also address
the power leakage problem. Finally, the simulation results and field
trials revealed the effectiveness of the CoTP-based G2A coverage extension
method.


\bibliographystyle{IEEEtran}
\bibliography{seamlesscoverage}


\end{document}